\documentclass[letterpaper, 10pt, conference]{ieeeconf}      

\IEEEoverridecommandlockouts                              
\usepackage{amsmath}    
\usepackage{amsfonts}
\usepackage{graphicx}   
\usepackage{subcaption}
\usepackage{epsfig} 
\usepackage{color}
\usepackage[normalem]{ulem}
\usepackage{cancel}
\usepackage{amssymb}
\usepackage{color}
\usepackage[ruled,vlined,titlenotnumbered]{algorithm2e} 

\newcommand{\R}{\mathbb{R}} 
\newcommand{\veh}{Q} 
\newcommand{\cset}{\mathcal{U}}
\newcommand{\cfset}{\mathbb{U}}

\newcommand{\targetset}{\mathcal{L}}
\newcommand{\reachset}{\mathcal{V}}
\newcommand{\dz}{\mathcal{Z}} 
\newcommand{\Ulo}{\hat{U}} 
\newcommand{\ulo}{\hat{u}} 
\newcommand{\Rcm}{C} 
\newcommand{\rcm}{c} 
\newcommand{\sv}{s} 
\newcommand{\st}{K}
\newcommand{\Prm}{P} 
\newcommand{\prm}{p} 

\newcommand{\pcs}{\mathcal{J}} 

\newtheorem{defn}{Definition}
\newtheorem{prop}{Proposition}
\newtheorem{thm}{Theorem}
\newtheorem{rmk}{Remark}
\newtheorem{corr}{Corollary}

\title{\LARGE \bf
Multi-Vehicle Collision Avoidance via Hamilton-Jacobi Reachability and Mixed Integer Programming}

\author{Mo Chen*, Jennifer C. Shih*, and Claire J. Tomlin
\thanks{This work has been supported in part by NSF under CPS:ActionWebs (CNS-931843), by ONR under the HUNT (N0014-08-0696) and SMARTS (N00014-09-1-1051) MURIs and by grant N00014-12-1-0609, by AFOSR under the CHASE MURI (FA9550-10-1-0567). The research of M. Chen has received funding from the ``NSERC PGS-D'' Program.}
\thanks{* Both authors contributed equally to this work. All authors are with the Department of Electrical Engineering and Computer Sciences, University of California, Berkeley. \{mochen72,cshih,tomlin\}@berkeley.edu}
}

\begin{document}
\maketitle
\thispagestyle{empty}
\pagestyle{empty}

\begin{abstract}
Multi-agent differential games are important and useful tools for analyzing many practical problems. With the recent surge of interest in using UAVs for civil purposes, the importance and urgency of developing tractable multi-agent analysis techniques that provide safety and performance guarantees is at an all-time high. Hamilton-Jacobi (HJ) reachability has successfully provided safety guarantees to small-scale systems and is flexible in terms of system dynamics. However, the exponential complexity scaling of HJ reachability prevents its direct application to large scale problems when the number of vehicles is greater than two. In this paper, we overcome the scalability limitations of HJ reachability by using a mixed integer program that exploits the properties of HJ solutions to provide higher-level control logic. Our proposed method provides safety guarantee for three-vehicle systems -- a previously intractable task for HJ reachability -- without incurring significant additional computation cost. Furthermore, our method is scalable beyond three vehicles and performs significantly better by several metrics than an extension of pairwise collision avoidance to multi-vehicle collision avoidance. We demonstrate our proposed method in simulations.
\end{abstract}

\section{Introduction}
From projects such as Amazon Prime Air and Google Project Wing to other recent uses of unmanned aerial vehicles (UAVs), there is without a doubt an immense interest in using UAVs for civil purposes \cite{Debusk10, AUVSI16, Amazon16, BBC16}. Potential uses of UAVs include package delivery, aerial surveillance, and disaster response \cite{Tice91}; future applications of UAVs are only limited by imagination. As a result, government agencies such as the Federal Aviation Administration (FAA) and the National Aeronautics and Space Administration (NASA) are urgently working on UAV-related regulations \cite{FAA13, NASA16, Kopardekar16}.

Much research has gone into the area of multi-agent systems, which involve aspects of cooperation and asymmetric goals among the agents. In \cite{Fiorini98, Vandenberg08}, the authors assume that the vehicles will employ certain simple control strategies which induce velocity obstacles that must be avoided in order to maintain safety. Other approaches involved using potential functions to ensure collision avoidance while multiple agents maintain formation to travel along pre-specified trajectories \cite{Saber02, Chuang07}. Although approaches like these provide valuable insight to multi-agent systems, they do not flexibly offer the safety guarantees that are desirable in safety-critical systems.

Multi-agent systems have also been studied in the context of differential games, which are ideal for addressing safety-critical problems such as the ones involving UAVs we now urgently face, because of the safety and performance guarantees that differential game approaches can provide. The HJ formulation of differential games has been studied extensively and successfully applied to small-scale problems involving one or two vehicles \cite{Vaisbord88, Mitchell05, Fisac15, Ding08}. Besides providing safety guarantees, perhaps the most appealing feature of HJ-based methods is its flexibility in terms of the system dynamics. Unfortunately, the computation complexity of HJ-based methods scales exponentially with the number of vehicles in the system, making their direct application to multi-vehicle problems intractable.

Many attempts have also been made to use differential games to analyze larger-scale problems. For example, in works such as \cite{Tanimoto78, Su14, Fisac15b}, the authors discuss various classes of three-player differential game with different assumptions on the role of each agent in non-cooperative settings. For even larger systems, \cite{Lin13, Chen14, Chen15, Chen15b} provide promising results when varying degrees of structural assumptions can be made. However, none of these attempts at providing guarantees address the problem of unstructured flight, which may be important in some situations. In addition, having stronger safety guarantees in unstructured environments has the potential to make structured flight of UAVs more resilient to unforeseen circumstances. 

In this paper, we build on the HJ-based method for guaranteeing safety when no more than two vehicles are present. We augment the HJ-based method with a higher-level joint cooperative control strategy using a mixed integer program (MIP) inspired by the properties of the pairwise safety guarantee. Our proposed MIP scales well with the number of vehicles, provides safety guarantees for three vehicles, and results in significantly better performance for multi-vehicle systems in general compared to when not using the higher-level control logic. We provide a proof for the safety guarantee in a three-vehicle system, and illustrate the safety guarantee and performance benefits through simulations of multi-vehicle systems in various configurations.


\section{Problem Formulation \label{sec:formulation}}
Consider $N$ vehicles, denoted $\veh_i, i = 1, 2, \ldots, N$, described by the following ordinary differential equation (ODE)
\vspace{-0.75em}
\begin{equation}
\label{eq:vdyn} 
\dot{x}_i = f_i(x_i, u_i), \quad u_i \in \cset_i, \quad i = 1,\ldots, N
\end{equation}

\noindent where $x_i \in \R^{n_i}$ is the state of the $i$th vehicle $\veh_i$, and $u_i$ is the control of $\veh_i$. Each of the $N$ vehicles may have some objective, such as getting to a set of goal states. Whatever the objective may be, each vehicle $\veh_i$ must at all times avoid the \textit{danger zone} $\dz_{ij}$ with respect to each of the other vehicles $\veh_j, j = 1, \ldots, N, j \neq i$. In general, the danger zones $\dz_{ij}$ may represent any relative configuration between $\veh_i$ and $\veh_j$ that are considered undesirable, such as collision. In this paper, we make the assumption that $x_i \in \dz_{ij} \Leftrightarrow x_j \in \dz_{ji}$, the interpretation of which is that between a pair of vehicles, an unsafe configuration is one in which either of the vehicle is the danger zone of the other.

If possible and desired, each vehicle would use a ``liveness controller'' that helps complete its objective. However, sometimes a ``safety controller'' must be used in order to prevent the vehicle from entering any danger zones with respect to any other vehicles. Since the danger zones $\dz_{ij}$ are sets of joint configurations, it is convenient to derive the set of relative dynamics between every vehicle pair from the dynamics of each vehicle specified in $\eqref{eq:vdyn}$. Let the relative dynamics between $\veh_i$ and $\veh_j$ be specified by the ODE
\vspace{-0.5em}
\begin{equation}
\label{eq:rdyn} 
\begin{aligned}
\dot{x}_{ij} &= g_{ij}(x_{ij}, u_i, u_j) \\
u_i &\in \cset_i, u_j \in \cset_j \quad i, j = 1, \ldots, N, i\neq j 
\end{aligned}
\end{equation}

We assume the functions $f_i$ and $g_{ij}$ are uniformly continuous, bounded, and Lipschitz continuous in arguments $x_i$ and $x_{ij}$ respectively for fixed $u_i$ and $(u_i, u_j)$ respectively. In addition, the control functions $u_i(\cdot)\in\cfset_i$ are drawn from the set of measurable functions\footnote{A function $f:X\to Y$ between two measurable spaces $(X,\Sigma_X)$ and $(Y,\Sigma_Y)$ is said to be measurable if the preimage of a measurable set in $Y$ is a measurable set in $X$, that is: $\forall V\in\Sigma_Y, f^{-1}(V)\in\Sigma_X$, with $\Sigma_X,\Sigma_Y$ $\sigma$-algebras on $X$,$Y$.}.

Given the vehicle dynamics in \eqref{eq:vdyn}, some joint objective, the derived relative dynamics in \eqref{eq:rdyn}, and the danger zones $\dz_{ij}, i,j = 1, \ldots, N, i \neq j$, we propose a cooperative safety control strategy that performs the following:

\begin{enumerate}
\item detect potential conflict based on the joint configuration of all $N$ vehicles;
\item allow vehicles that are not in potential conflict to complete their objective using a liveness controller;
\item among the vehicles in potential conflict, attempt to minimize the number of instances in which a vehicle gets into another vehicle's danger zone.
\end{enumerate}

For the case of $N=3$, we prove that our proposed control strategy guarantees that all vehicles will be able to stay out of all the danger zones with respect to the other vehicles, and thus guaranteeing safety. For all initial configurations in our simulations, all vehicles also complete their objectives.


\section{Methodology \label{eq:method}}
Our proposed method builds on HJ reachability theory, which in the case of $N=2$ guarantees no vehicle will enter another vehicle's danger zone and that the vehicles will eventually complete their joint objective \cite{Mitchell05}. HJ reachability becomes computationally intractable for $N>2$. To provide the same guarantees for $N=3$, we propose an MIP motivated by the properties of the HJ pairwise solution to specify a higher level control logic. While unable to provide hard guarantees for $N>3$, our proposed method is computationally tractable for much larger $N$, and performs significantly better than applying an extension of the pairwise HJ reachability solution when $N>3$.

\subsection{Hamilton-Jacobi Reachability \label{sec:HJI}}
HJ reachability has been studied extensively \cite{Mitchell05, Fisac15, Barron90, Bokanowski10, Margellos11} and found many successful applications \cite{Mitchell05, Ding08, Chen14, Margellos13}. Here, we give a brief overview of how to apply HJ reachability to solve a pairwise collision avoidance problem such as the one in \cite{Mitchell05}. Given the relative dynamics \eqref{eq:rdyn}, we define the target set to be the danger zone $\dz_{ij}$, and compute following the backward reachable set
\vspace{-1em}
%
%
%
%

\begin{equation}
\label{eq:brs}
\begin{aligned}
&\reachset_{ij}(t) = \{x_{ij}: \forall u_i \in \cfset_i, \exists u_j \in \cfset_j, \\
&x_{ij}(\cdot) \text{ satisfies \eqref{eq:rdyn}}, \exists s \in [0, t], x_{ij}(s) \in \dz_{ij}\}
\end{aligned}
\end{equation}

If $x_{ij}$, the relative state of $\veh_i$ and $\veh_j$ is outside of $\reachset_{ij}$ for all $j$, then $\veh_i$ is free to use a liveness controller to make progress towards its objective. If $x_{ij}$ is on the boundary of $\reachset_{ij}$ for a single $j$, then danger can be averted, regardless of the action of $\veh_j$, by using the optimal control denoted $u_{ij}^*$, which can be obtained from the gradient of the value function $V_{ij}(t, x_{ij})$ representing $\reachset_{ij}(t)$. For details on obtaining $V_{ij}$, see \cite{Mitchell05}; for this paper, it is sufficient to note that $\reachset_{ij} = \{x_{ij}: V_{ij}(t, x_{ij}) \le 0\}$ where we assume $t \rightarrow \infty$ and write $V_{ij}(x_{ij}) = \lim_{t \rightarrow \infty} V_{ij}(t, x_{ij})$. The interpretation is that $\veh_i$ is guaranteed to be able to avoid collision with $\veh_j$ over an infinite time horizon as long as the optimal control $u_{ij}^*$ is applied as soon as the potential conflict occurs.


If $x_{ij}$ is in $\reachset_{ij}$ for more than one $j$, then the pairwise optimal controls $u_{ij}^*$ cannot guarantee safety. However, in this case, our proposed cooperative control strategy, which uses a MIP to provide a higher level control logic, can provide safety guarantees when $N=3$.

%
%
%

\subsection{The Mixed Integer Program \label{sec:MIP}}
For the $N > 2$ case, we use an MIP to provide higher level control logic to synthesize a cooperative safety controller. We first note two properties of the pairwise solution:

\begin{enumerate}
\item If every vehicle pair stays out of each other's danger zones, then the entire set of $N$ vehicles would be out of each other's danger zones.
\item Since the solution is pairwise, the safety controller derived from HJ reachability can only guarantee that some vehicle $i$ can avoid the danger zone with respect to a single other vehicle $j$.
\end{enumerate}

Intuitively, a higher level control logic is needed to provide a far-sighted avoidance maneuver; without this higher level logic, pairwise avoidance maneuvers between two vehicles $\veh_i$ and $\veh_j$ may lead to unavoidable dangerous configurations with respect to a third vehicle $\veh_k$. 

\begin{defn}
\textbf{Control logic matrix}: Let $\Ulo \in \{0, 1\}^{N \times N}$ be the control logic matrix specifying the joint cooperative control of the $N$ vehicles. Denote the element of $\Ulo$ in position $(i, j)$ to be $\ulo_{ij}$. If $\ulo_{ij} = 1$, then the control logic stipulates that vehicle $\veh_i$ must execute the pairwise optimal control to avoid vehicle $\veh_j$.
\end{defn}

\begin{defn}
\textbf{Reward coefficient matrix}: Let $\Rcm \in \R^{N \times N}$ be the reward coefficient matrix with elements $\rcm_{ij}$. Each $\rcm_{ij}$ specifies the ``reward'' for choosing to have vehicle $i$ avoiding vehicle $j$, or in other words, choosing $\ulo_{ij} = 1$.
\end{defn}

Motivated by the above two properties, and using the above definitions, we arrive at the following MIP:
\vspace{-0.75em}
\begin{equation}
\label{eq:baseMIP}
\begin{aligned}
\max_{\ulo_{ij}} & \sum_{i, j} c_{ij} \ulo_{ij} \\
\text{subject to } & \ulo_{ij} + \ulo_{ji} \le 1 & \forall i, j, i \neq j & \quad (\ref{eq:baseMIP}a) \\
& \sum_j \ulo_{ij} \le 1 & \forall i & \quad (\ref{eq:baseMIP}b) \\
& \ulo_{ij} \in \{0, 1\} & \forall i, j, i \neq j & \quad (\ref{eq:baseMIP}c)
\end{aligned}
\vspace{-0.5em}
\end{equation}

At a given time, the vehicles' joint state determines $\Rcm$, which forms the objective of \eqref{eq:baseMIP}. Thus, the interpretation of the objective of \eqref{eq:baseMIP} depends on the choice of the reward coefficient matrix $\Rcm$. A large $\rcm_{ij}$ encourages $\ulo_{ij}$ to be $1$, causing vehicle $\veh_i$ to avoid $\veh_j$. The decision variables consist of the elements of $\Ulo$, which provides the high level control logic. This is captured by constraint (\ref{eq:baseMIP}c). 

The pairwise HJ optimal control guarantees that a vehicle $\veh_i$ can remain safe with respect to another vehicle $\veh_j$ regardless of the action of $Q_j$. Therefore, in every pair $(\veh_i, \veh_j)$, if either $\veh_i$ or $\veh_j$ is avoiding the other, there is no need for the other vehicle to also be avoiding the first. The constraint (\ref{eq:baseMIP}a) states that out of every vehicle pair, at most one vehicle should avoid the other so that no control authority is wasted by having both vehicles avoid each other. The other vehicle then could use its control authority to avoid a third vehicle with whom it may come into conflict.

Finally, since the control logic ultimately results in vehicles performing pairwise optimal controls, each vehicle is only guaranteed to be able to avoid at most one other vehicle. The constraint (\ref{eq:baseMIP}b) encodes this limitation.

\subsection{Design of the Objective Function \label{sec:Rcm}}
The objective function in \eqref{eq:baseMIP} can be designed by choosing the reward coefficient matrix $\Rcm$. In general, there may be many choices for $\Rcm$, and the general guiding principle in choosing $\Rcm$ is that it should depend on the vehicles' safety levels and avoidance priority; both concepts are defined below. In this paper, we propose one particular choice of $\Rcm$ that allows us to prove safety guarantees for three vehicles.

Given the form of the objective function, the first obvious choice for some of the elements of $\Rcm$ would be $\rcm_{ii} = -\infty, \forall i$. This forces $\ulo_{ii} = 0~\forall i$, which states that a vehicle $\veh_i$ does not need to avoid itself. Before designing the rest of $\Rcm$, we need to define the notion of a safety level.

\begin{defn}
\textbf{Safety level}: Given $x_{ij}$, the state of vehicle $\veh_i$ with respect to vehicle $\veh_j$, define the safety level to be $V_{ij}(x_{ij})$. For convenience, let $\sv_{ij} = V_{ij}(x_{ij})$.
\end{defn}

\begin{prop}
Suppose $\sv_{ij} > 0$ at some time $t = t_0$. If $\veh_i$ chooses the control  $u_{ij}^*$, then $\sv_{ij} > 0 ~ \forall t > t_0$.
\end{prop}

\begin{proof}
Based on the definitions of $\sv_{ij}$, $\reachset_{ij}$, and $V_{ij}(x_{ij})$, we have that if $\sv_{ij} > 0$ at $t = t_0$, then the control $u_{ij}^*$ collision avoidance for an infinite time horizon. This implies $\sv_{ij} > 0$ for all time.
\end{proof}

\begin{corr}
\label{corr:pw_safety}
Between the pair $(\veh_i, \veh_j)$, if $\sv_{ij} > 0$ or $\sv_{ji} > 0$, then there exists a joint control strategy $(u_i, u_j)$ to ensure neither vehicle enters the danger zone of the other.
\end{corr}

\begin{proof}
If $\sv_{ij} > 0$, then safety is guaranteed if $\veh_i$ chooses $u_i = u_{ij}^*$ to avoid $\veh_j$. If $\sv_{ij} \le 0$, then $\sv_{ji} > 0$. In this case, simply swap the indices $i$ and $j$.
\end{proof}

Let $\st$ be a safety level threshold. We say $\veh_i$ is in \textit{potential conflict} with $\veh_j$ if $s_{ij} < \st$. Based on this safety level threshold, we set $\rcm_{ij} = -1$ whenever $\sv_{ij} > \st$. So far, we have $\rcm_{ij} = -\infty$ whenever $i = j$ and $\rcm_{ij} = -1$ whenever $\sv_{ij} > \st$. The rest of the values of $\Rcm$ are derived from the priority matrix, defined below.

\begin{defn}
\textbf{Priority matrix}: Let $\Prm \in \{1, 2, \ldots, N^2-N\}^{N \times N}$ be a priority matrix with elements $\prm_{ij}$. The priority matrix establishes an avoidance order for the vehicles.

The diagonal elements of $\Prm$ can be arbitrarily set (denoted $*$). The rest of the elements are assigned in descending order according to Sarrus' rule \cite{Khattar10} (for determining cross products). For example, in the case of $N = 3$,
\vspace{-0.5em}
\begin{equation}
\begin{aligned}
&\Prm =
\begin{bmatrix}
* & 6 & 3 \\
2 & * & 5 \\
4 & 1 & *
\end{bmatrix}
\end{aligned}
\end{equation}
\end{defn}

A large value of $\prm_{ij}$ indicates that $\veh_i$ should avoid $\veh_j$ with a high priority. In order to impose such a priority when constructing a joint cooperative safety control strategy, we set $\rcm_{ij} = \prm_{ij}^2$ whenever $\sv_{ij} \le K$. For example, if $N = 3$ and $\sv_{ij} \le K ~ \forall i, j$, then we would have
\vspace{-0.5em}
\begin{equation}
\vspace{-0.5em}
\label{eq:Rcmfull}
\Rcm =
\begin{bmatrix}
-\infty & 36 & 9 \\
4 & -\infty & 25 \\
16 & 1 & -\infty
\end{bmatrix}
\end{equation}

As another example, if $N = 3, \sv_{ij} \le K ~ \forall i, j$ except $\sv_{13}, \sv_{32} > K$, then we would have
\vspace{-0.75em}
\begin{equation}
\vspace{-0.5em}
\Rcm =
\begin{bmatrix}
-\infty & 36 & -1 \\
4 & -\infty & 25 \\
16 & -1 & -\infty
\end{bmatrix}
\end{equation}

\begin{rmk}
Avoidance priority is an important notion for guaranteeing safety even when $N = 2$. Consider the scenario where vehicle $\veh_i$ applies the control $u_i^*$ to avoid $\veh_j$, but $\veh_j$ does not try to avoid $\veh_i$. As long as $\veh_i$ continues to avoid $\veh_j$, the two vehicles can avoid each other's danger zones.

While $\veh_i$ is avoiding $\veh_j$, $\sv_{ij}$ is guaranteed to remain positive; however, since $\veh_j$ is not avoiding $\veh_i$, $\sv_{ji}$ could become negative. If $\sv_{ji} < 0$, safety can \textit{only} be guaranteed if $\veh_i$ keeps avoiding $\veh_j$. The avoidance priority ensures that some $\veh_j$ never tries to avoid $\veh_i$ when $\sv_{ji} < 0$. Instead, the responsibility of avoidance would remain with $\veh_i$, which continues to avoid $\veh_j$ to ensure $\sv_{ij} > 0$.
\end{rmk}


\section{Safety Guarantee For Three Vehicles}
The method for constructing a joint safety controller described in Section \ref{eq:method} guarantees safety when $N = 3$. We now formally states this guarantee and prove the result.

\begin{thm}
\label{thm:main_result}
Suppose $N = 3$. If $\sv_{12}, \sv_{23}, \sv_{31} \ge 0$ at some time $t = t_0$, then the joint control strategy from the MIP \eqref{eq:baseMIP} with the reward coefficient matrix elements $\rcm_{ij}$ chosen in Section \ref{sec:Rcm} guarantees that $\sv_{12}, \sv_{23}, \sv_{31} \ge 0 ~ \forall t > t_0$.
\end{thm}

\begin{proof}
It suffices to show that $0 \le \sv_{12}, \sv_{23}, \sv_{31} \le \st$ at $t = t_0$ implies $\sv_{12}, \sv_{23}, \sv_{31} \ge 0 ~ \forall t > t_0$.

Suppose $0 \le \sv_{12} \le \st$. Then $\rcm_{12} = 36$. From the objective of \eqref{eq:baseMIP}, $\ulo_{12}$ would be chosen to be $1$ unless another feasible solution in which $\ulo_{12} = 0$ results in a higher objective value. Due to (\ref{eq:baseMIP}a) and (\ref{eq:baseMIP}b), the only way for the optimal solution to have $\ulo_{12} = 0 $ is to have $\ulo_{13} = 1$ or $\ulo_{21} = 1$.

There are several cases of $\Rcm$ to go through, with each case having different elements of $\Rcm$ being equal to $-1$. We show one case here; the rest of the cases follow a similar logic. Assume $\Rcm$ is given in \eqref{eq:Rcmfull}. Then, since $\rcm_{ij} > 0 ~ \forall i,j,i\neq j$, the optimal solution would have as many elements of $\Ulo$ being $1$ as possible (except for diagonal elements).

Suppose $\ulo_{21} = 1$, then by (\ref{eq:baseMIP}a), $\ulo_{12} = 0$ and by (\ref{eq:baseMIP}b), $\ulo_{23} = 0$. This leaves us with the freedom to choose $\ulo_{13}, \ulo_{31}, \ulo_{32}$. Since $\Rcm_{31} = 16 > 9 + 1 = \Rcm_{13} + \Rcm_{32}$, choosing $\ulo_{31} = 1, \ulo_{13} = \ulo_{32} = 0$ would maximize the objective. This gives us the candidate solution $\ulo_{21}=1, \ulo_{31} = 1$ and the rest of the $\ulo_{ij}$ being $0$, with an objective value of $4 + 16 = 20$. However, choosing $\ulo_{12} = 1$ alone would already result in an objective value of at least $\Rcm_{12} = 36$; therefore, $\ulo_{21} \neq 1$.

Next, suppose $\ulo_{13} = 1$, then by (\ref{eq:baseMIP}a), $\ulo_{12} = 0$ and by (\ref{eq:baseMIP}b), $\ulo_{31} = 0$. This leaves us with the freedom to choose $\ulo_{21}, \ulo_{23}, \ulo_{32}$. Since $\Rcm_{23} = 25 > 4 + 1 = \Rcm_{21} + \Rcm_{32}$, choosing $\ulo_{23} = 1, \ulo_{21} = \ulo_{32} = 0$ would maximize the objective. This gives us the candidate solution $\ulo_{13}=1, \ulo_{23} = 1$ and the rest of the $\ulo_{ij}$ being $0$, with an objective value of $9 + 25 = 34$. However, choosing $\ulo_{12} = 1$ alone would already result in an objective value of at least $\Rcm_{12} = 36$; therefore, $\ulo_{13} \neq 1$.

This leaves us with $\ulo_{12} = 1$ whenever $0 \le \sv_{12} \le \st$. By a similar argument, one can show that $\ulo_{23} = 1$ whenever $0 \le \sv_{23} \le \st$, and $\ulo_{31} = 1$ whenever $0 \le \sv_{31} \le \st$.
\end{proof}

\begin{rmk}
Alternatively, one could enumerate all feasible solutions for every possible choice of $\Rcm$, and discover the same result stated in Theorem \ref{thm:main_result}. We have also taken this brute force approach to verify the above proof.
\end{rmk}

\begin{corr}
By Theorem \ref{thm:main_result} and Corollary \ref{corr:pw_safety}, if $N = 3$ and each vehicle $\veh_i$ uses the optimal pairwise safety controller $u_{ij}^*$ with respect to $\veh_j$ whenever $\ulo_{ij} = 1$, then no vehicle will ever get into another vehicle's danger zone.
\end{corr}




\section{Numerical Simulations}
In this section, we illustrate our proposed method through simulations and compare our method with a baseline pairwise method that uses solely the HJ pairwise optimal control solution in which each agent $Q_{i}$ avoids the agent $Q_{j}$ in the potential conflict set $\pcs_{i}$ with the smallest pairwise safety value $s_{ij}$. Compared with our MIP formulation (\ref{eq:baseMIP}), the baseline can be thought of as a different MIP that
\begin{itemize}
  \item omits constraint ($\ref{eq:baseMIP}$a), making the vehicles unable to coordinate among each other, and
  \item assumes $\forall i, c_{ij}=1$ if $\veh_i$ has the lowest safety value with respect to $\veh_j$, and $c_{ij}=-\infty$ otherwise, making the vehicles lack a notion of global avoidance priority.
\end{itemize}

Such a baseline is chosen to illustrate the benefits of the above design considerations, which are important features of our proposed method. For illustration purposes, we assumed that the dynamics of each vehicle is given by
\vspace{-0.5em}
\begin{equation}
\label{eq:dyn_i}
\begin{aligned}
\dot{p}_{x,i} &= v \cos \theta_i \\
\dot{p}_{y,i} &= v \sin \theta_i \\
\dot{\theta}_i &= \omega_i, \quad |\omega_i| \le \bar{\omega}
\end{aligned}
\vspace{-0.5em}
\end{equation}

\noindent where the state variables $p_{x,i}, p_{y,i}, \theta_i$ represent the $x$ position, $y$ position, and heading of vehicle $\veh_i$. Each vehicle travels at a constant speed of $v=5$, and chooses its turn rate $\omega_i$, constrained by some maximum $\bar{\omega} = 1$. The danger zone for HJ computation between $\veh_i$ and $\veh_j$ is defined as
\begin{equation}
\targetset_{ij} = \{x_{ij}: (p_{x,i} - p_{x,j})^2 + (p_{y,i} - p_{y,j})^2 \le R_c^2\},
\end{equation}

\noindent whose interpretation is that $\veh_i$ and $\veh_j$ are considered to be in each other's danger zone if their positions are within $R_c$ of each other. In our examples, we chose $R_c = 5$. Here, $x_{ij}$ represents their joint state, $x_{ij} = (p_{x,i} - p_{x,j}, p_{y,i} - p_{y,j}, \theta_i - \theta_j)$ For notational convenience, we define $p_{x, ij} \equiv p_{x,i}-p_{x,j}$, $p_{y, ij} \equiv p_{y,i} - p_{y, j}$, and $\theta_{ij} \equiv \theta_i - \theta_j$.

To obtain safety levels and the optimal pairwise safety controller, we compute the BRS \eqref{eq:brs} with the relative dynamics
\vspace{-0.5em}
\begin{equation}
  \label{eq:dyn_ij}
  \begin{aligned}
  \dot{p}_{x, ij} &=  -v + v \cos \theta_{ij} + \omega_i p_{y, ij} \\
  \dot{p}_{y, ij} &= v \sin \theta_{ij} - \omega_i p_{x, ij} \\
  \dot{\theta}_{ij} &= \omega_j - \omega_i, \quad |\omega_i|, |\omega_j| \le \bar{\omega}
  \end{aligned}
  \vspace{-0.5em}
\end{equation}


In our examples, we chose $\st=1.5$. Whenever $\sv_{ij} > \st ~ \forall j$, $\veh_i$ applies the optimal control to reach its destination\footnote{This optimal control can be computed by solving a reachability problem using the dynamics \eqref{eq:dyn_i}, but for brevity we will not go into the details here.}. Otherwise, $\veh_i$ uses the control specified by the joint cooperative safety controller that we propose in this paper.

Simulations for $N=3$ and $N=8$ are presented in detail for our method and the baseline method. Each vehicle aims to reach the circular target of matching color while avoiding other vehicles' danger zones. The vehicles keep traveling at constant speed even if they enter the danger zones of other vehicles until they reach their targets. The $s_{ij} = 0, \st$ safety level sets are plotted for some pairs of vehicles. When a vehicle is inside the $\st$ safety level set (outer boundary), plotted in the same color as the vehicle, it is in potential conflict with the vehicle around which the level set is plotted. However, as long as the vehicle stays outside of the $0$ safety level set (inner boundary), the pair of vehicles will be able to avoid entering each other's danger zones.

Fig. \ref{fig:our_3} illustrates how our joint collision avoidance method cooperatively resolves conflicts for three vehicles. The vehicles start outside of each others' $\st$ safety level sets. Each of them performs optimal control to reach their respective targets. On the way, $Q_2$ (green) and $Q_3$ (blue) come in conflict with each other. Cooperatively, $Q_2$ avoids $Q_3$ while $Q_3$ heads to the target since $Q_2$ is already resolving the pairwise conflict. At time $t=0.8$, all vehicles come in conflict with each other, and our proposed algorithm advises that $Q_{1}$ (red) avoids $Q_{2}$, $Q_{2}$ avoids $Q_{3}$, and $Q_3$ avoids $Q_1$, efficiently utilizing their control authorities for avoidance. At time $t=1.5$, the conflicts are resolved as each vehicle's safety level rises to above $\st=1.5$ with respect to the others. Eventually, all vehicles reach their targets without any entering each other's danger zones.

Fig. \ref{fig:naive_3} illustrates the pitfall of using the baseline method. Here, each vehicle avoids the vehicle with the smallest pairwise safety value. At $t = 0.6$, all vehicles come in conflict with each other, and without higher level logic, $Q_1$ (red) avoids $Q_3$ (blue), $Q_2$ (green) avoids $Q_1$, and $Q_3$ avoids $Q_1$. By avoiding each other, $Q_1$ and $Q_3$ waste control authority that can be used to prevent $Q_2$ and $Q_3$ from going closer to each other. When $Q_2$ and $Q_3$ come closer to each other, they begin avoiding each other, leading to $Q_1$ and $Q_3$  coming closer to each other. The lack of coordination causes this behavior to repeat, bringing them closer and closer together ($t=0.9$), and eventually leading them into each other's danger zones at $t=1.6$. This alternating avoidance behavior also highlights the importance of imposing avoidance priority.

Fig. \ref{fig:our_8} illustrates a difficult eight-vehicle scenario that our cooperative algorithm successfully resolves. The safety level sets are plotted for each avoidance pair. At $t = 2.7$, multiple vehicles are in conflict with each other. Notice that no redundant control is used (a pair of vehicles avoiding each other). Instead one vehicle in a given conflict pair can free up its control to avoid another agent. Fig. \ref{fig:naive_8} shows the result of applying the baseline approach, which is unable to resolve the multiple conflicts. In particular, at $t = 1.7$ (top right), multiple vehicle pairs avoid each other during the conflicts. In addition, at $t = 11.5$ (bottom right), two vehicles end up in a ``limbo'' state where they alternate between avoiding each other and trying to get closer to their targets, continually going in a direction that is \textit{further} from their targets.

\begin{figure}[]
\centering
  \begin{subfigure}[b]{0.19\textwidth}
    \includegraphics[width=\textwidth]{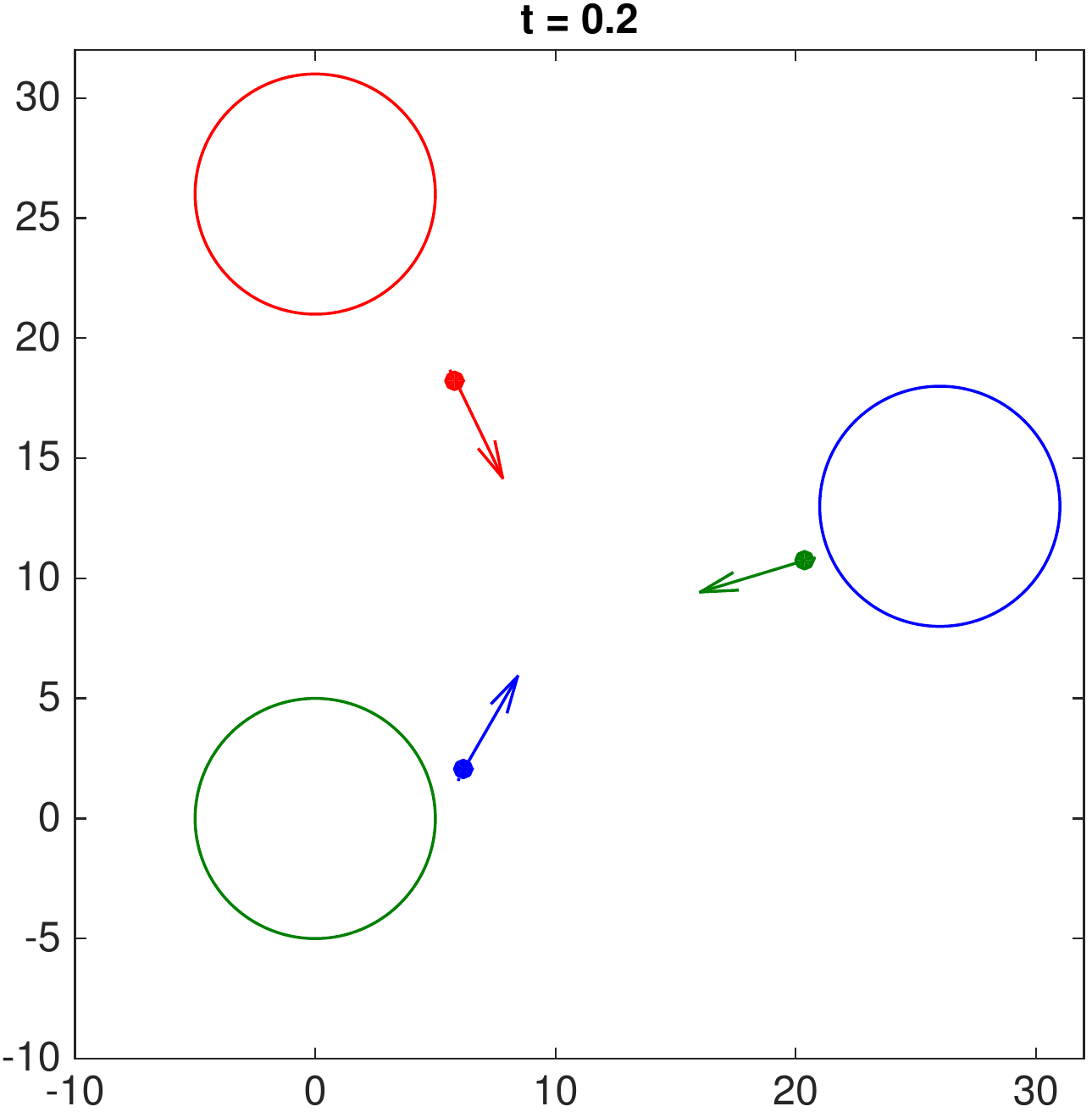}
  \end{subfigure}
  \begin{subfigure}[b]{0.19\textwidth}
    \includegraphics[width=\textwidth]{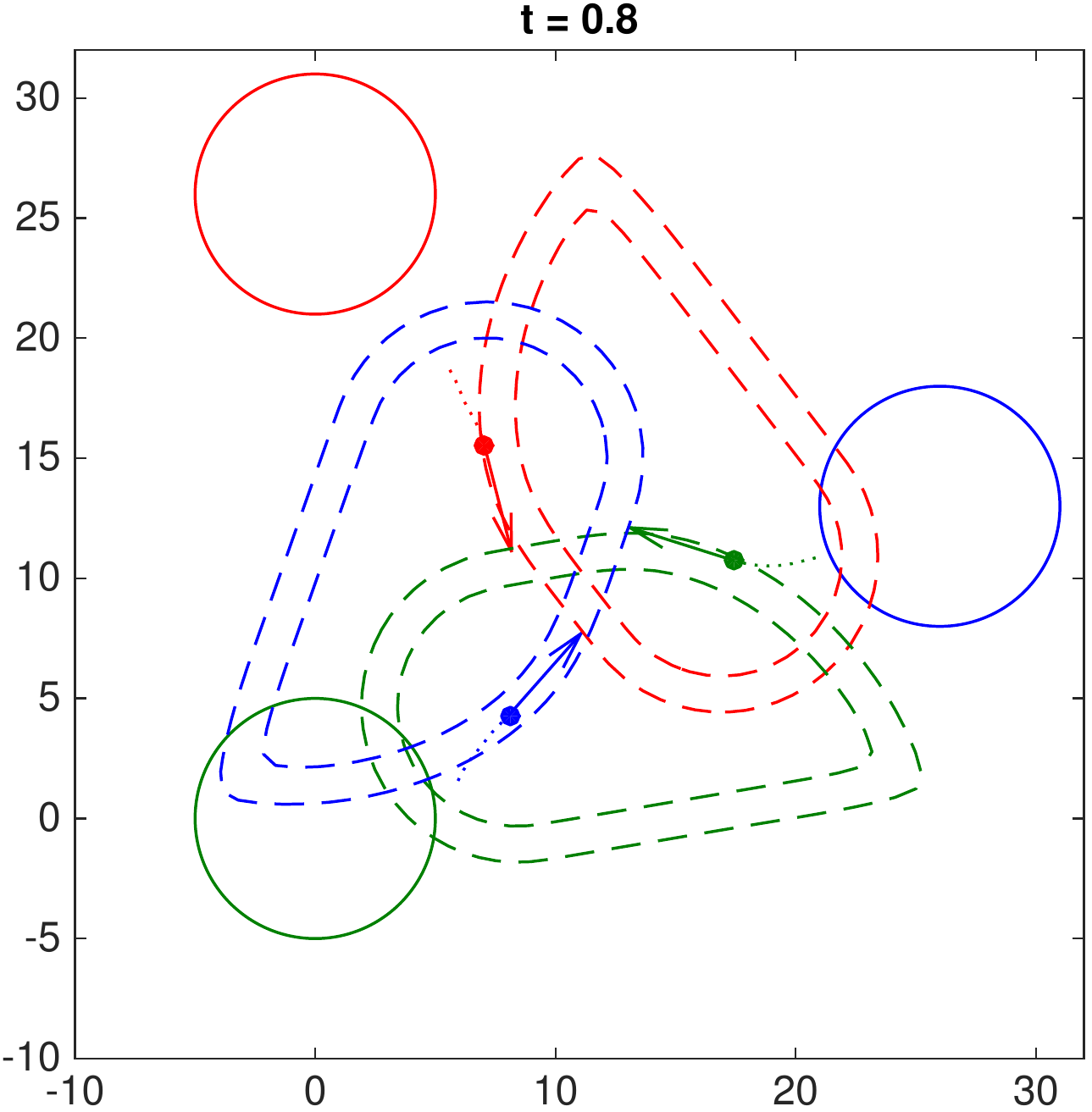}
  \end{subfigure}
  \\
  \begin{subfigure}[b]{0.19\textwidth}
    \includegraphics[width=\textwidth]{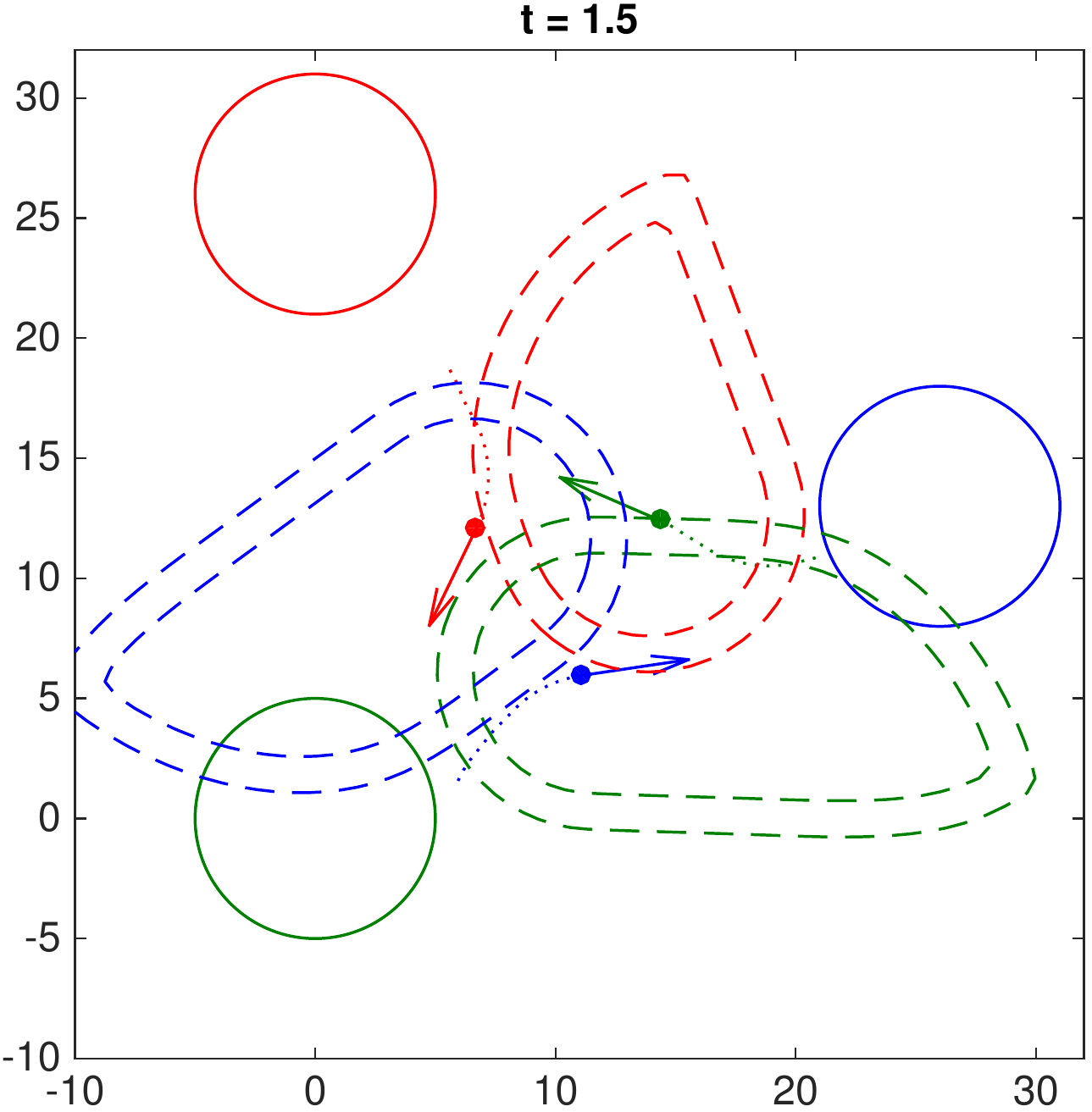}
  \end{subfigure}
  \begin{subfigure}[b]{0.19\textwidth}
    \includegraphics[width=\textwidth]{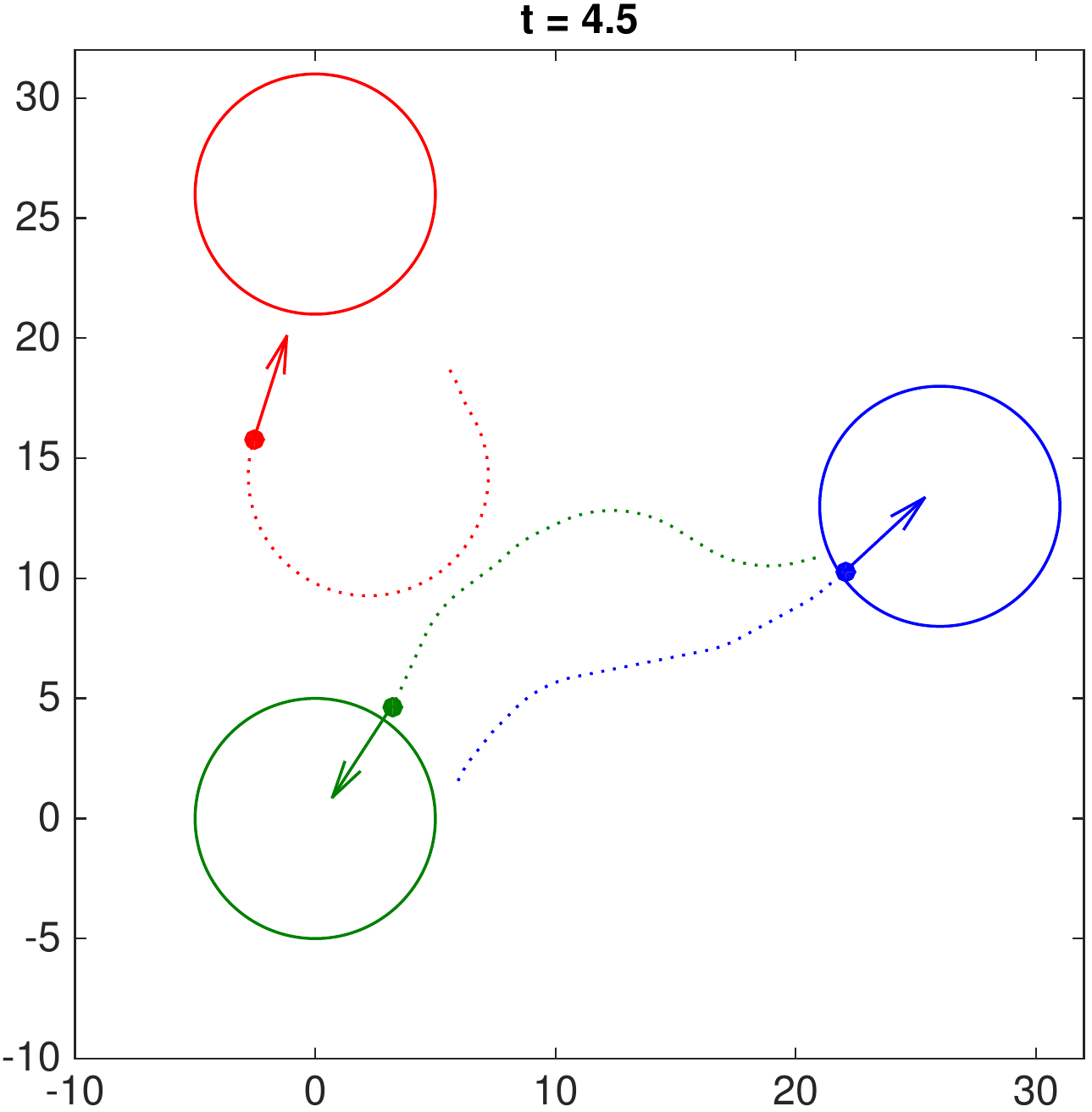}
  \end{subfigure}
  \\
  \begin{subfigure}[b]{0.15\textwidth}
    \includegraphics[width=\textwidth]{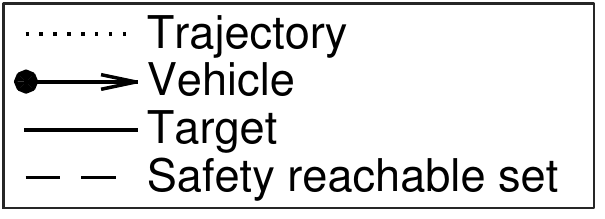}
  \end{subfigure}
  \caption{Three vehicles cooperatively resolve conflicts in a cyclic order, $Q_1$ (red) avoids $Q_2$ (green), $Q_2$ (green) avoids $Q_3$ (blue), and $Q_3$ (blue) avoids $Q_1$ (red).}
  \label{fig:our_3}
  \vspace{-2em}
\end{figure}

\begin{figure}[]
\centering
  \begin{subfigure}[b]{0.19\textwidth}
    \includegraphics[width=\textwidth]{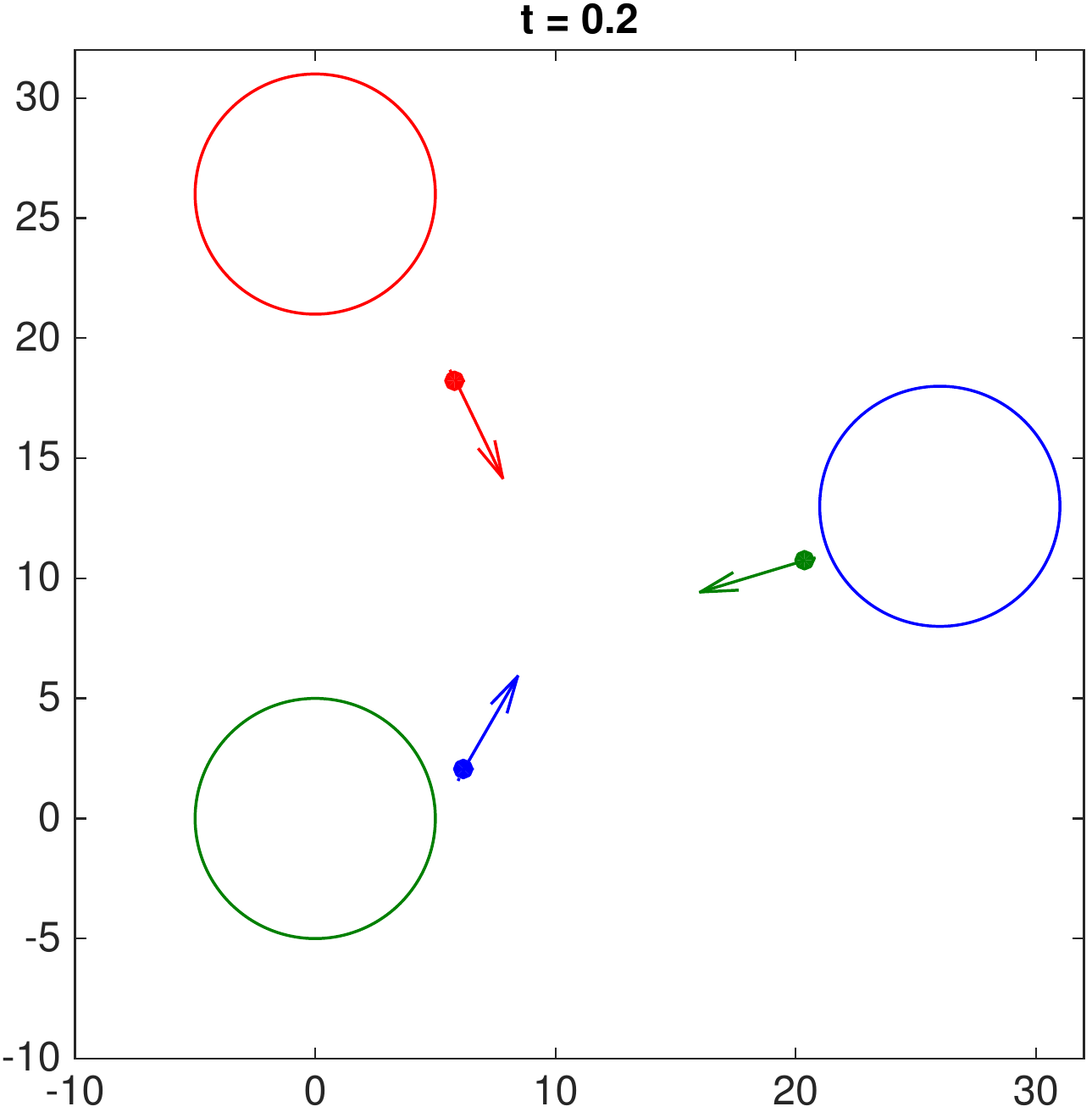}
  \end{subfigure}
  \begin{subfigure}[b]{0.19\textwidth}
    \includegraphics[width=\textwidth]{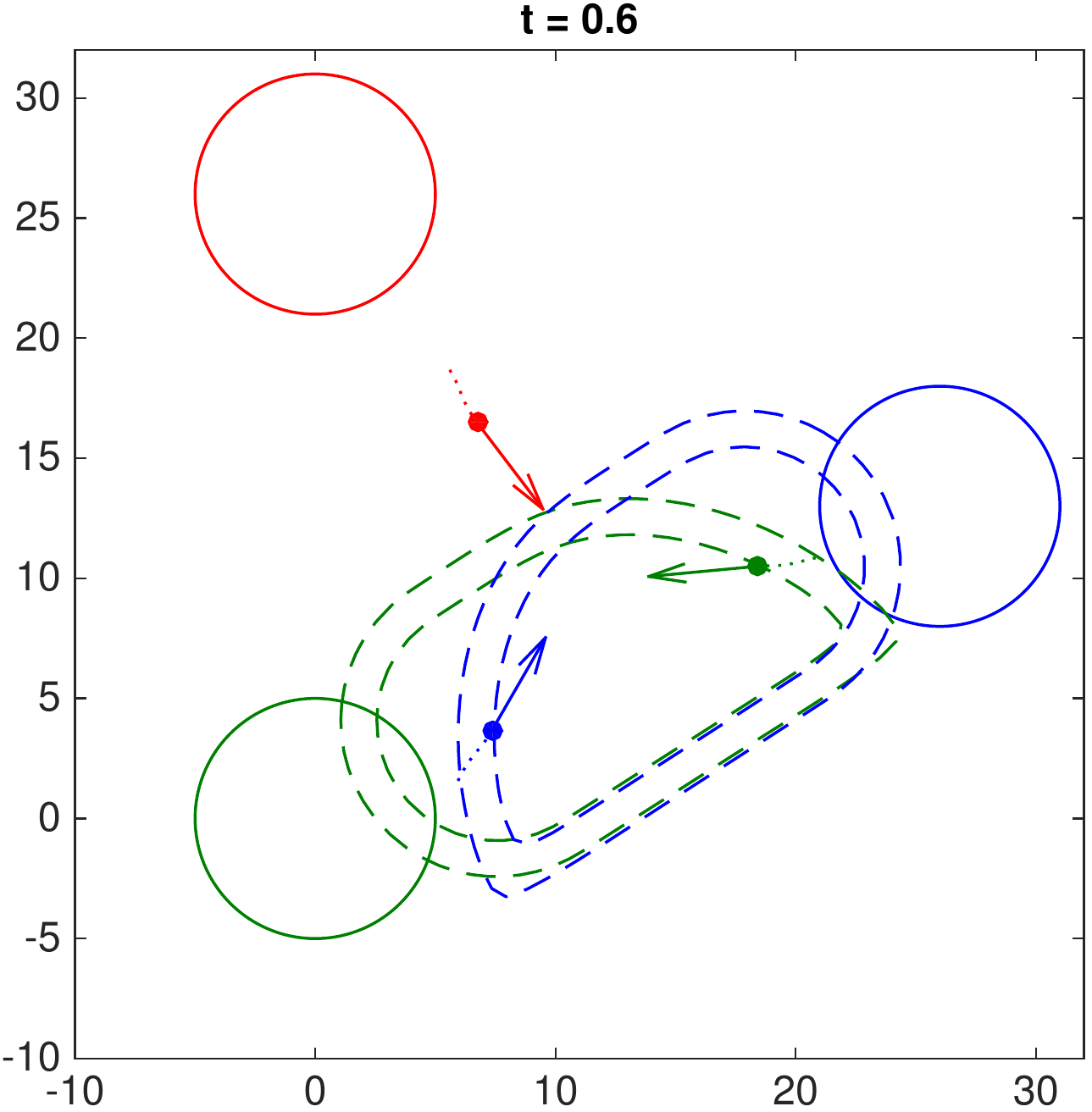}
  \end{subfigure}
  \\
  \begin{subfigure}[b]{0.19\textwidth}
    \includegraphics[width=\textwidth]{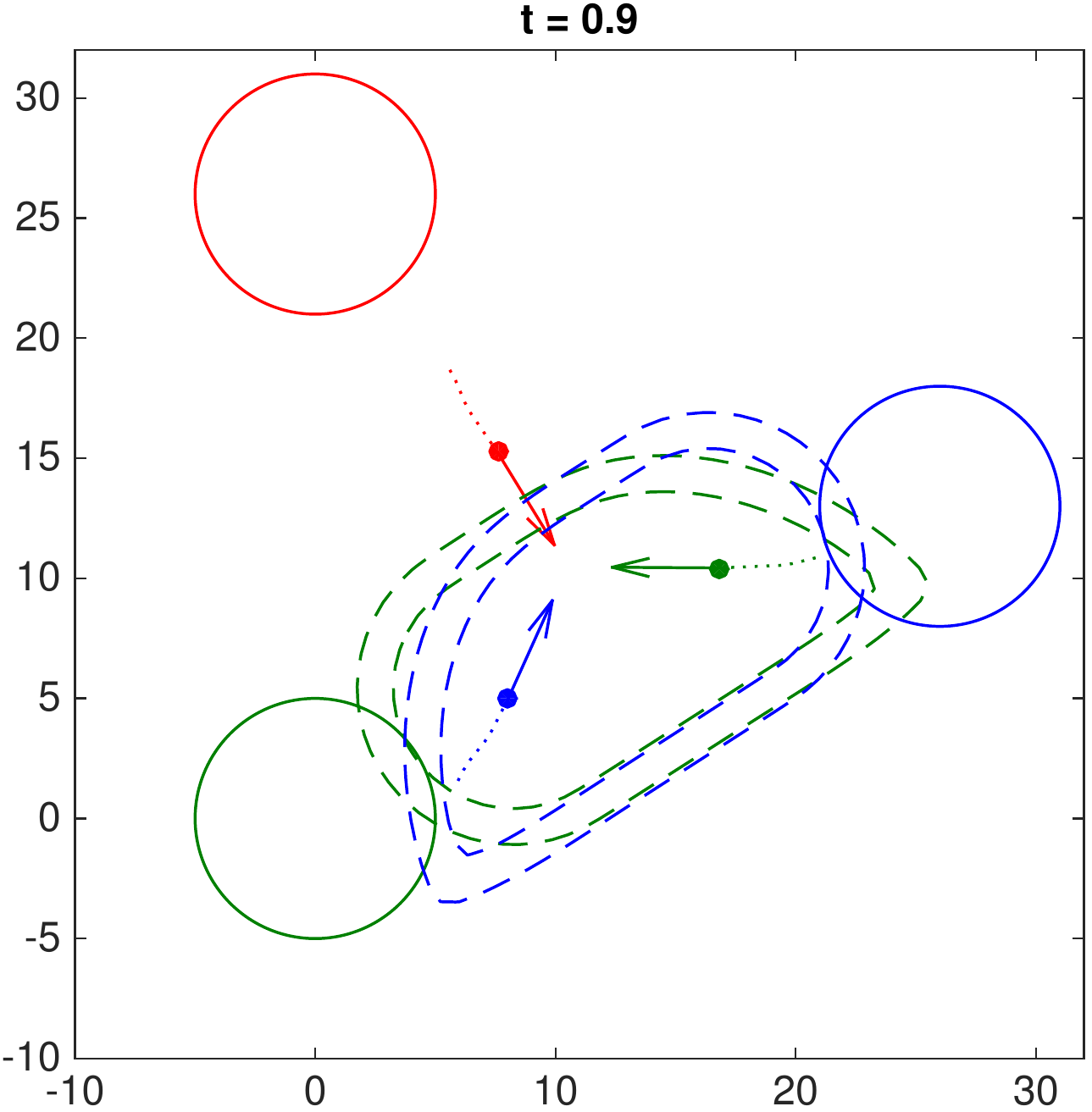}
  \end{subfigure}
  \begin{subfigure}[b]{0.19\textwidth}
    \includegraphics[width=\textwidth]{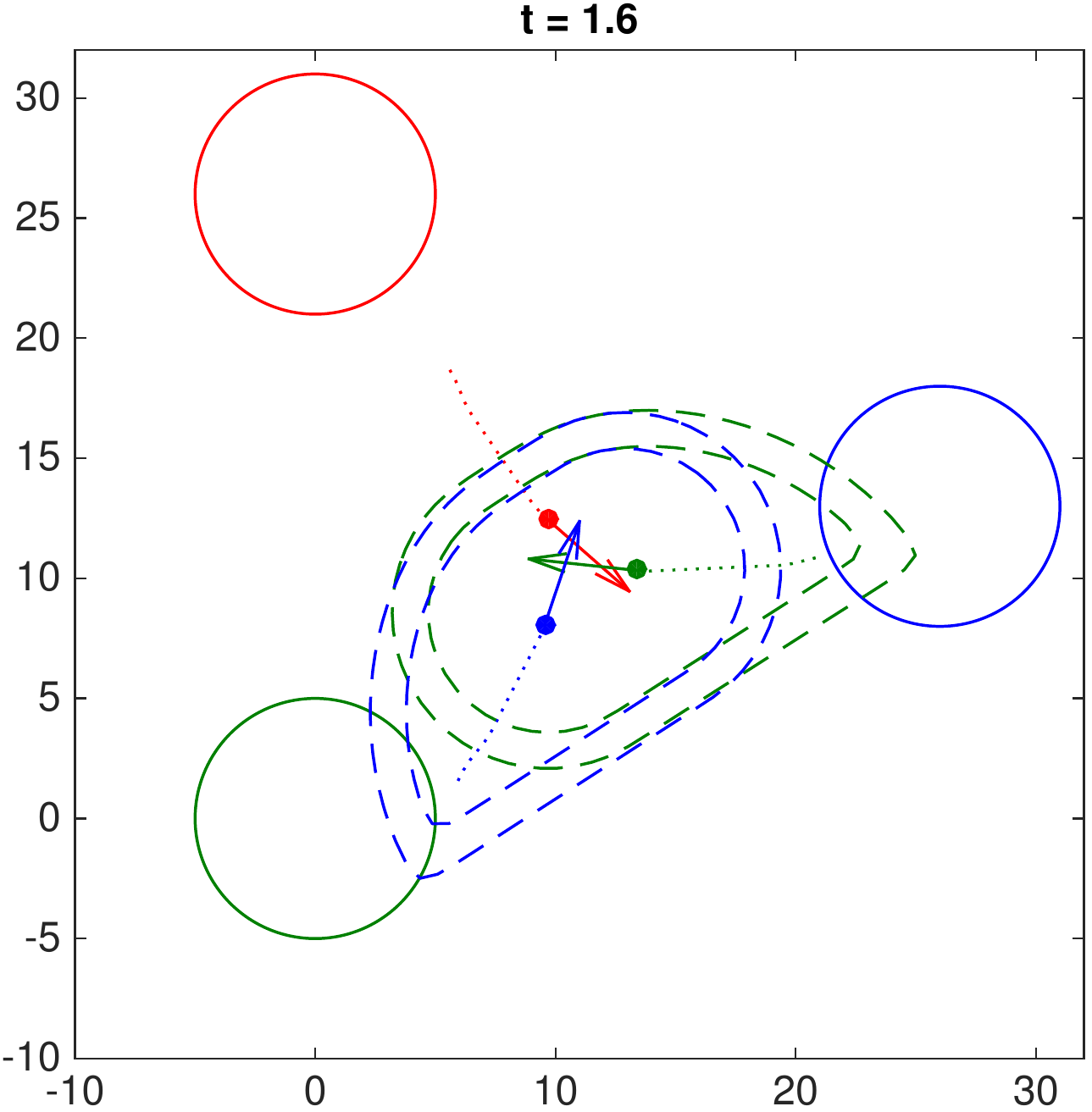}
  \end{subfigure}
  \caption{Without higher level control logic, the three vehicles are unable to resolve conflicts successfully. }
  \label{fig:naive_3}
  \vspace{-1em}
\end{figure}

\begin{figure}[]
\centering
  \begin{subfigure}[b]{0.21\textwidth}
    \includegraphics[width=\textwidth]{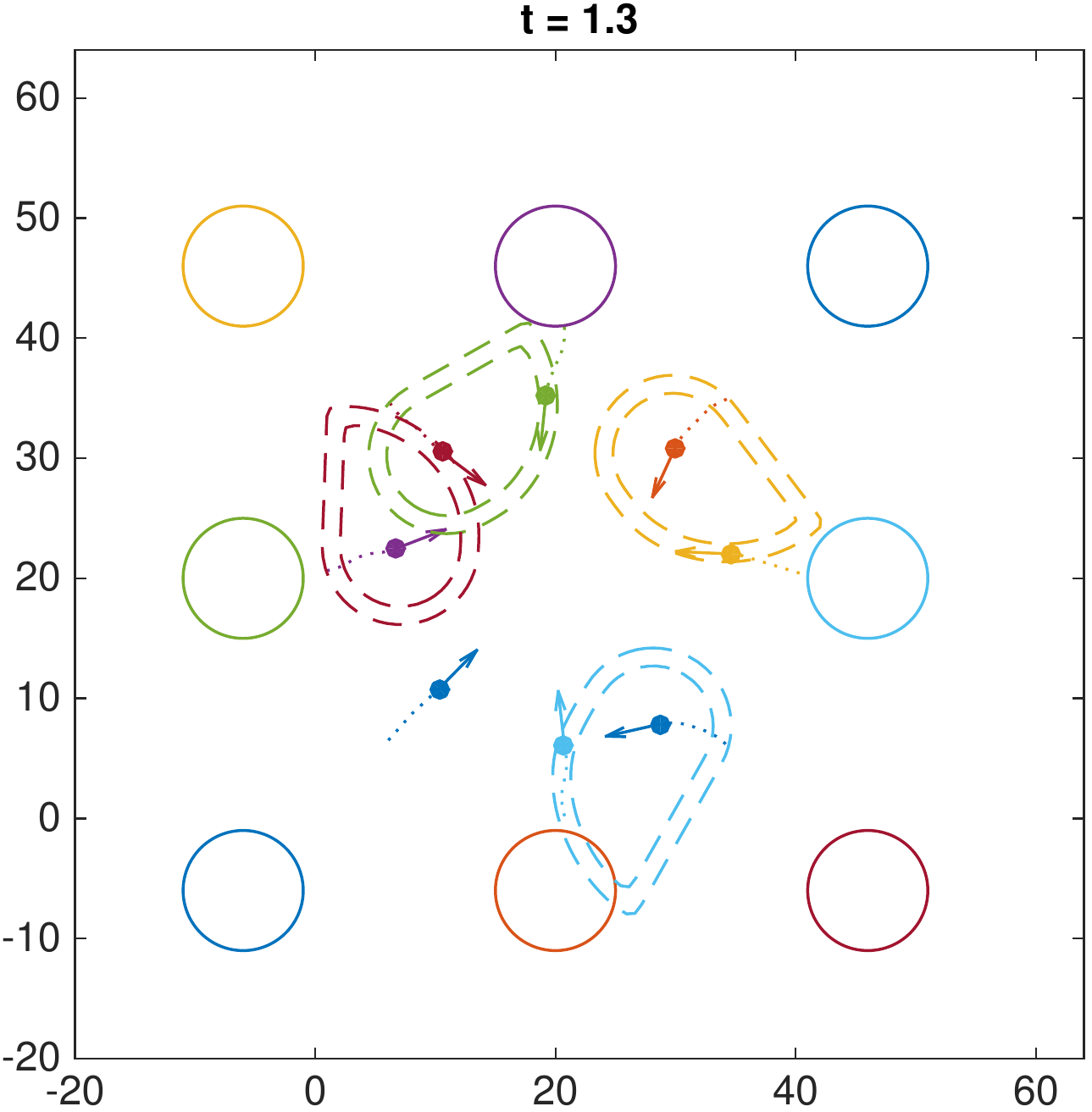}
  \end{subfigure}
  \begin{subfigure}[b]{0.21\textwidth}
    \includegraphics[width=\textwidth]{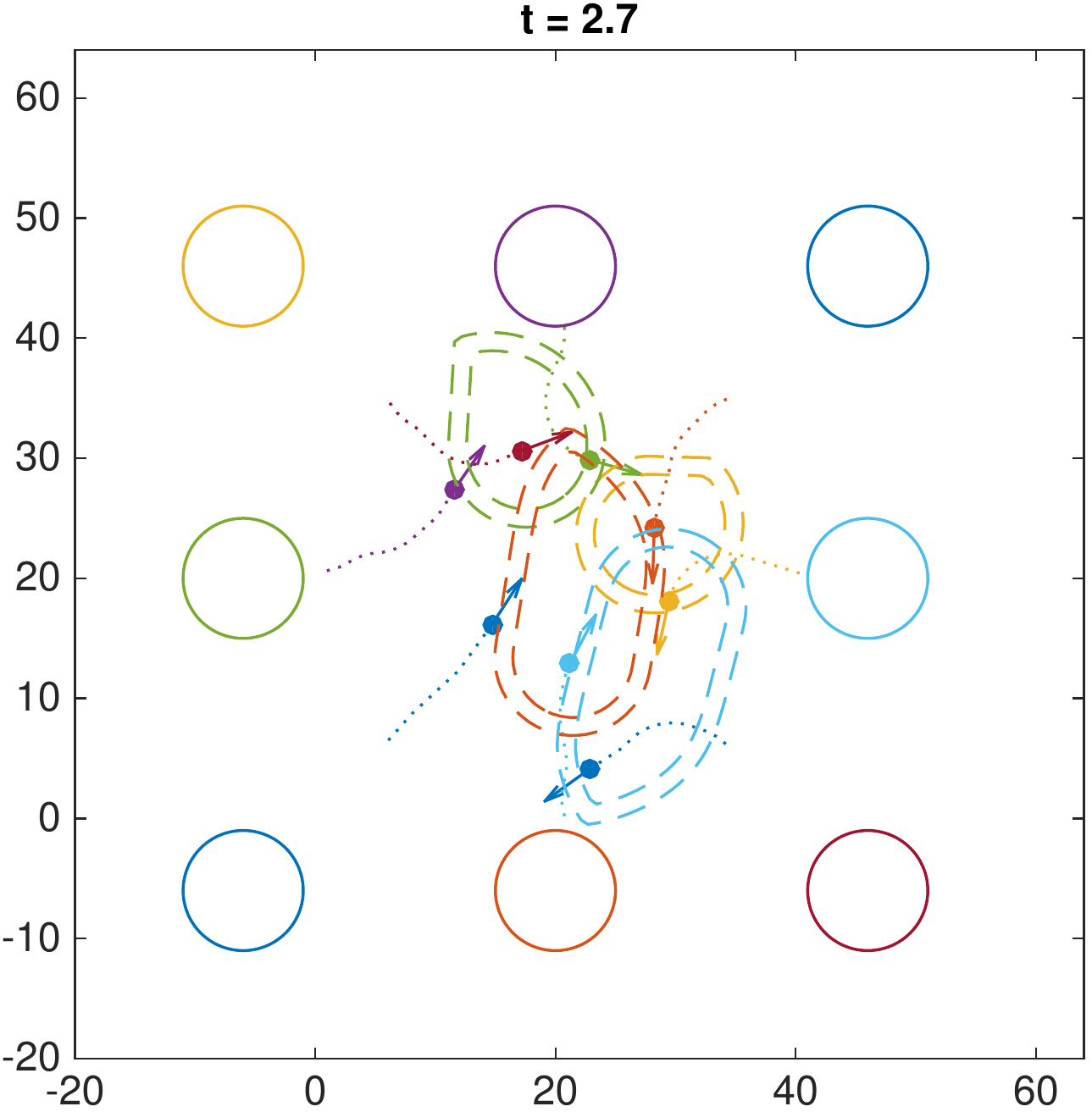}
  \end{subfigure}
  \\
  \begin{subfigure}[b]{0.21\textwidth}
    \includegraphics[width=\textwidth]{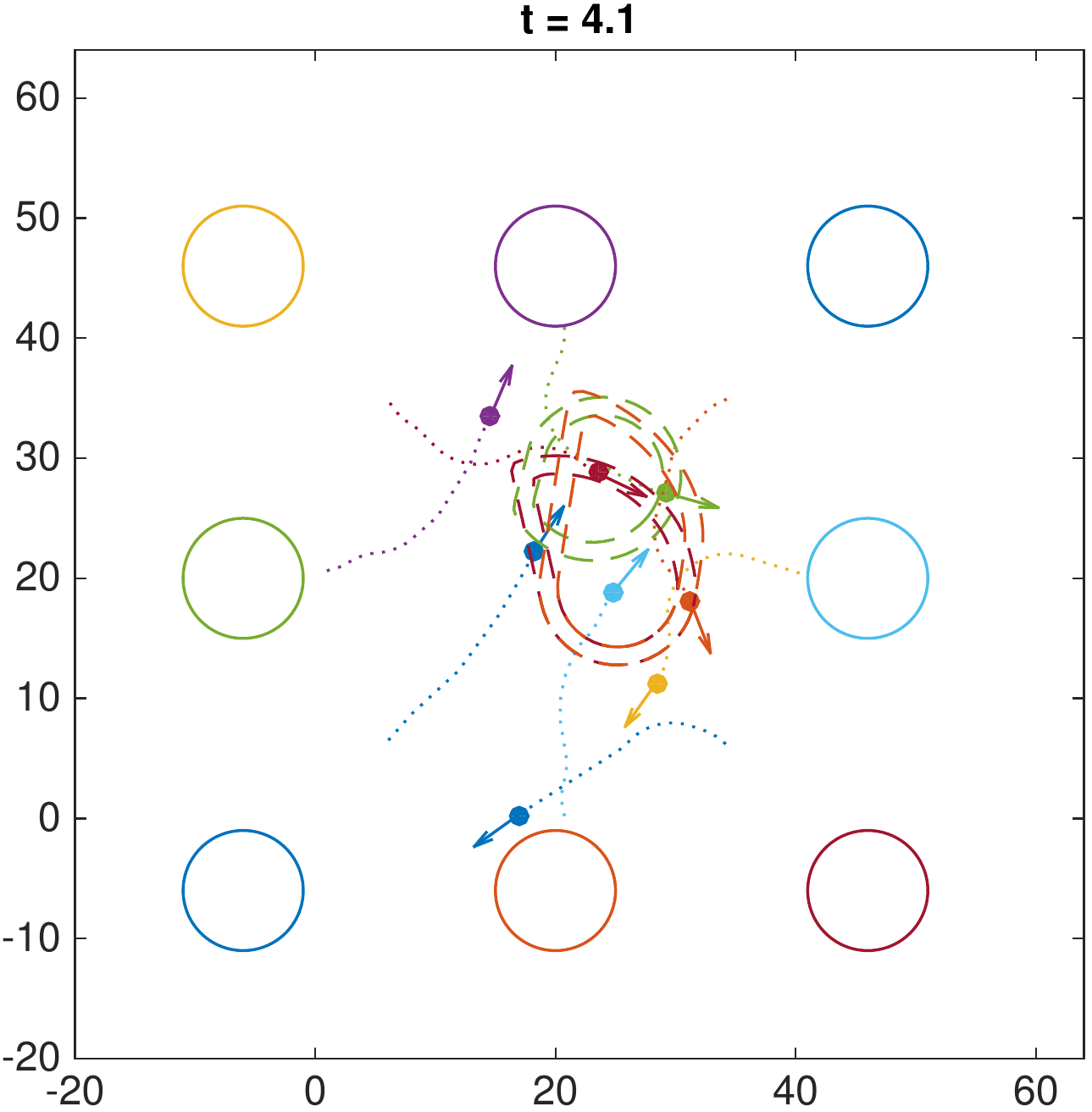}
  \end{subfigure}
  \begin{subfigure}[b]{0.21\textwidth}
    \includegraphics[width=\textwidth]{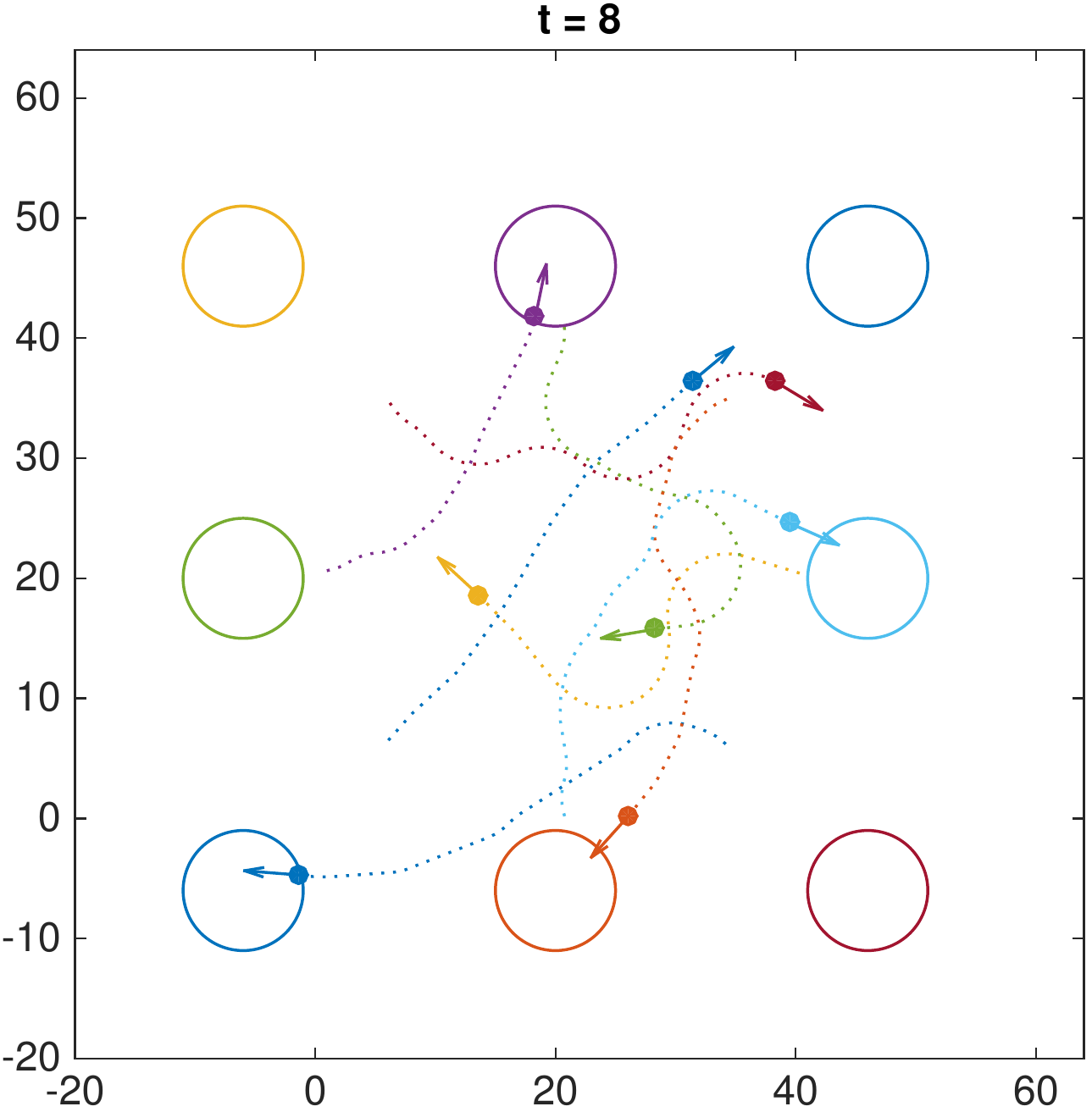}
  \end{subfigure}
  \caption{Eight vehicles successfully coordinated to resolve conflicts with our algorithm in this challenging scenario.}
  \label{fig:our_8}
\vspace{-2em}
\end{figure}

\begin{figure}[]
\centering
  \begin{subfigure}[b]{0.21\textwidth}
    \includegraphics[width=\textwidth]{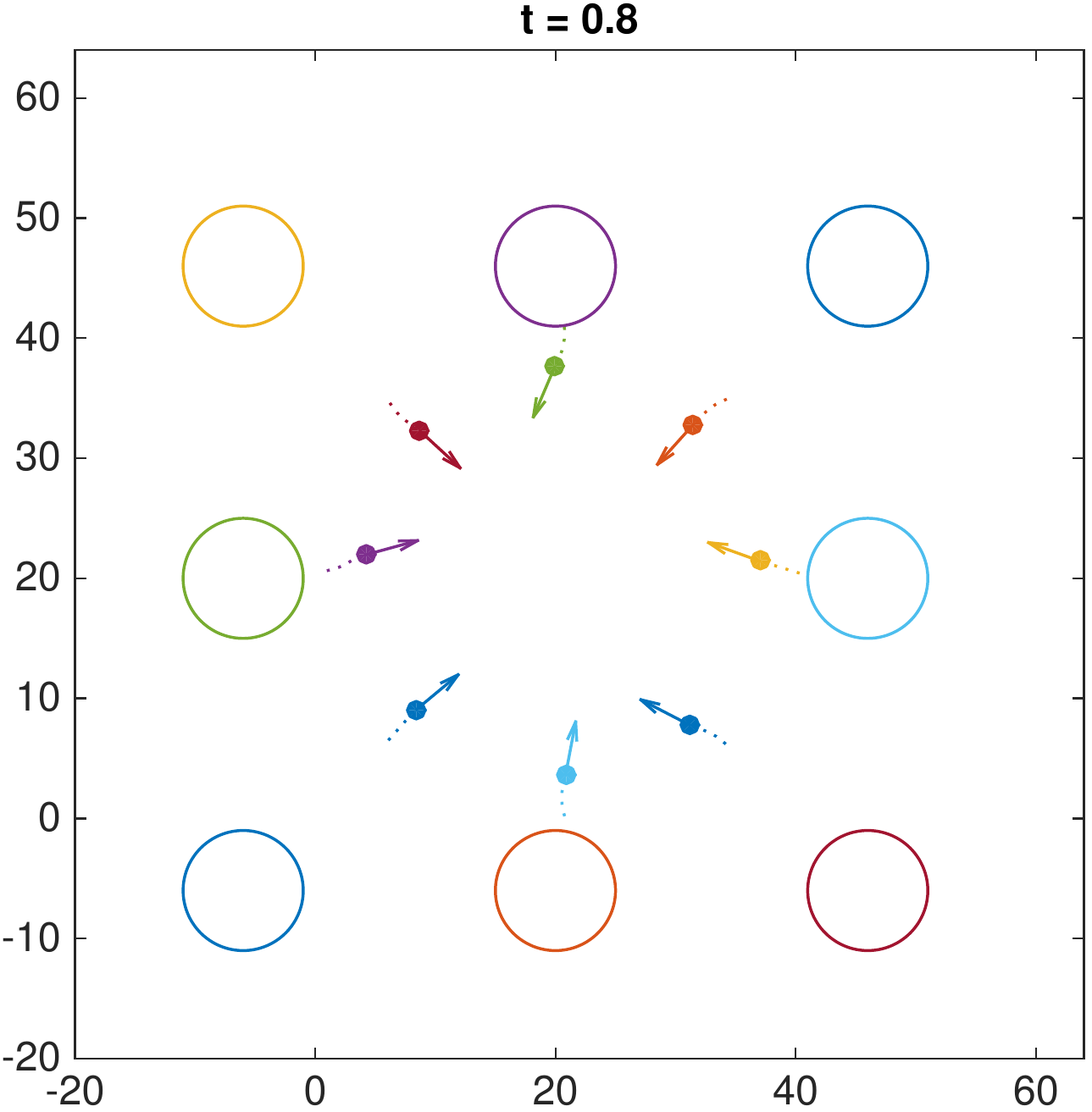}
  \end{subfigure}
  \begin{subfigure}[b]{0.21\textwidth}
    \includegraphics[width=\textwidth]{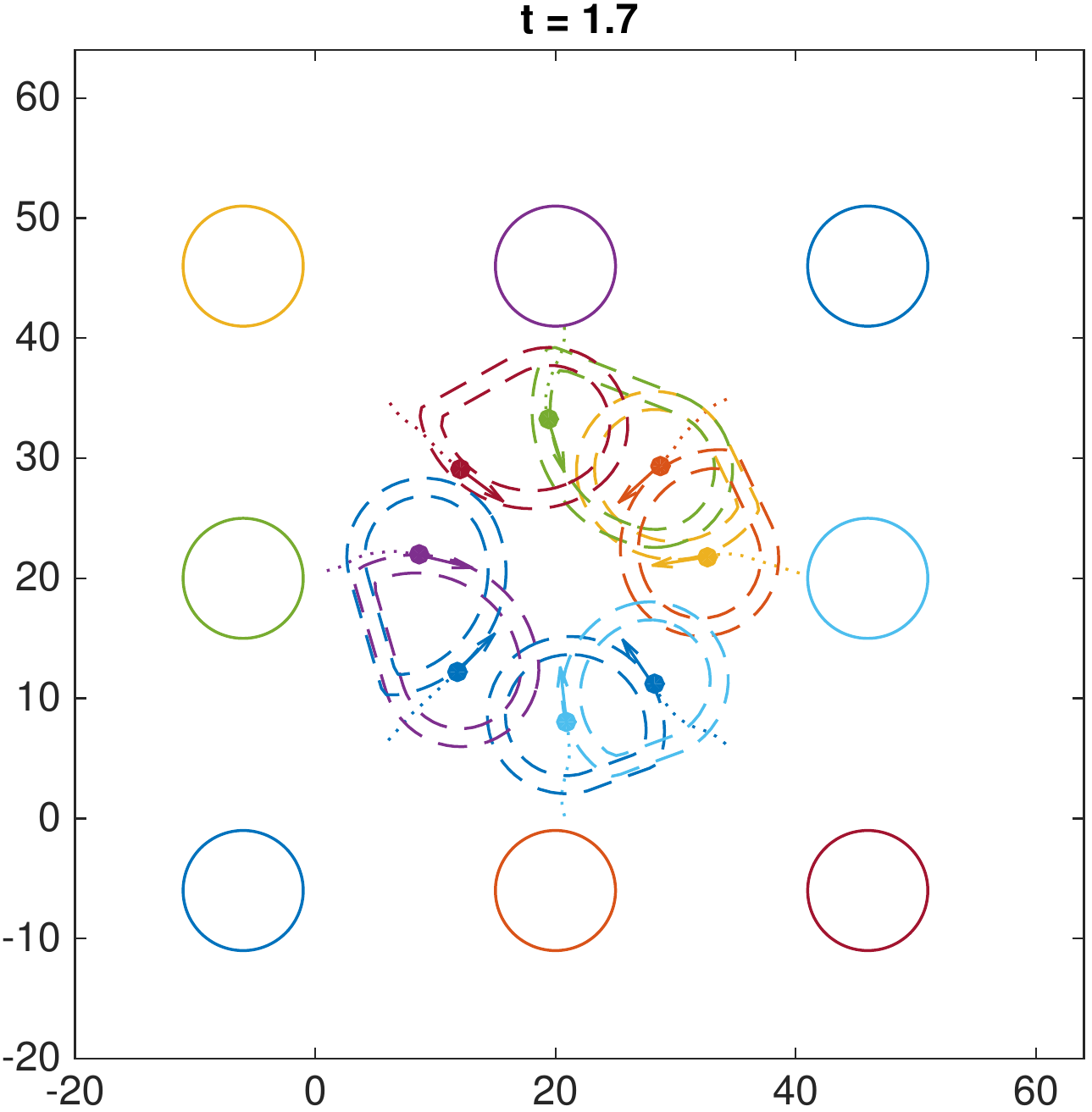}
  \end{subfigure}
  \\
  \begin{subfigure}[b]{0.21\textwidth}
    \includegraphics[width=\textwidth]{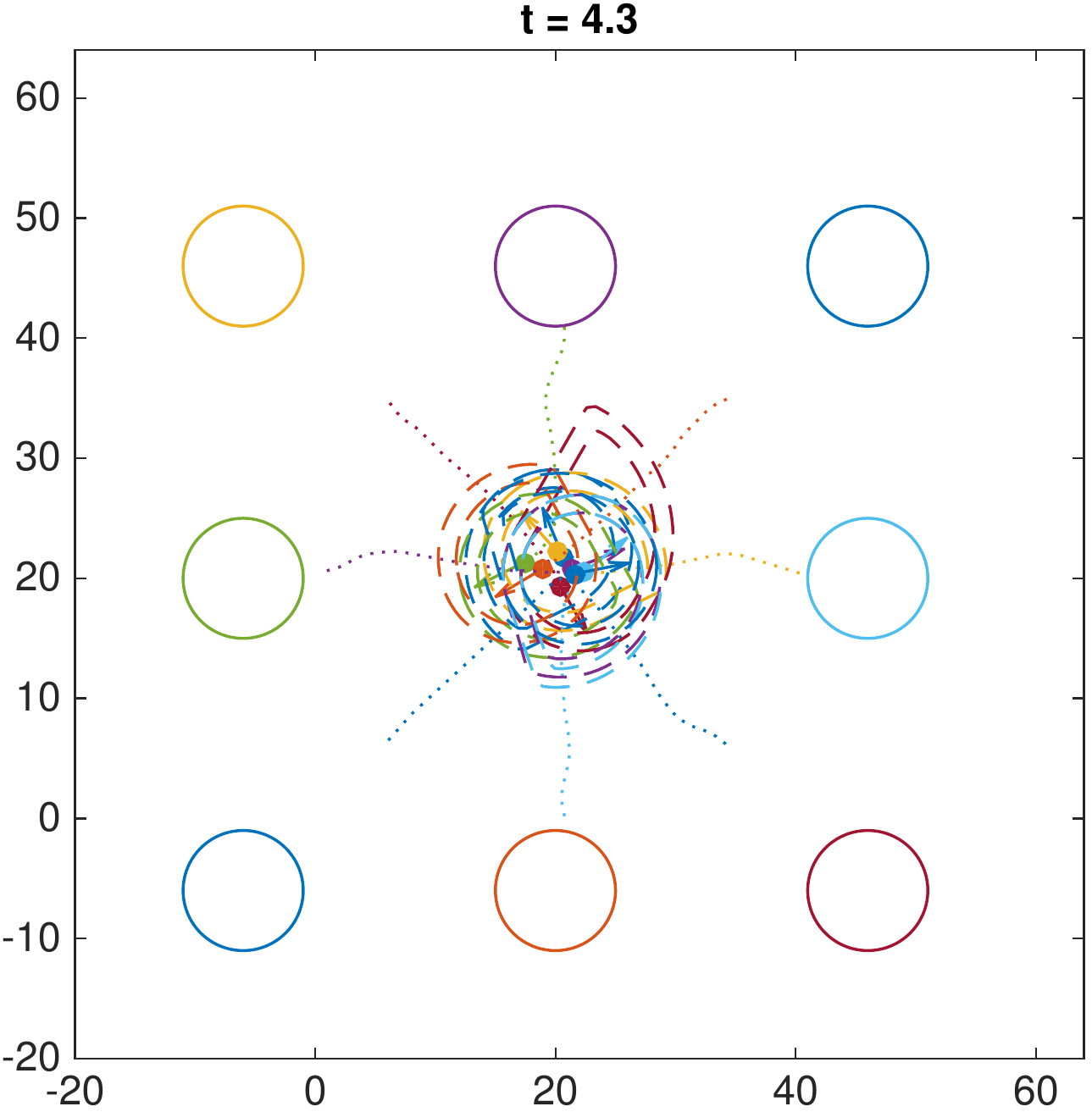}
  \end{subfigure}
  \begin{subfigure}[b]{0.21\textwidth}
    \includegraphics[width=\textwidth]{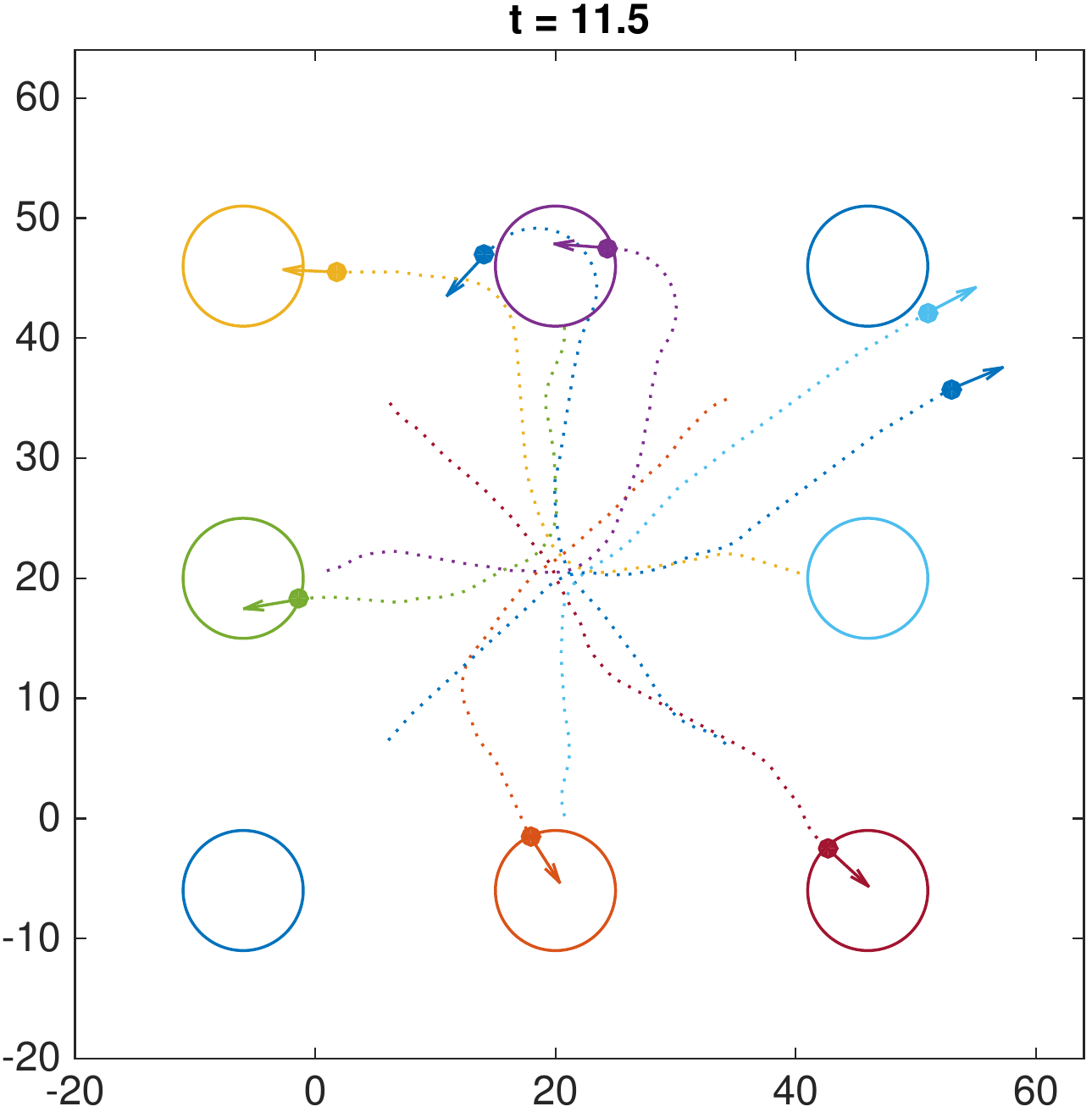}
  \end{subfigure}
  \caption{The lack of coordination using the baseline method results in failure in this challenging eight-vehicle scenario.}
  \label{fig:naive_8}
  \vspace{-1em}
\end{figure}

Additionally, we compare our method with the baseline method for $N=3,4,5,6,7,8$ vehicles 
by performing $200$ simulations with randomized initial conditions for each case, and show that our algorithm performs significantly better than the baseline pairwise approach. 
We initialized each vehicle by placing each of them symmetrically on a circle of radius $10 + 2 \times (N-3)$ facing the center of the circle, and then adding random perturbations to its initial state. We define the two performance metrics below. The average over the 200 trials for each case are presented in Fig. \ref{fig:simulation_combined}.

\begin{itemize}
  \item Success ratio = fraction of vehicles that reach their targets without ever entering others' danger zones
\item Aggregate conflict ratio = $\frac{\text{total \# of danger zone violations}}{\text{\# of time steps} \times C^{N}_{2}}$. The denominator is the maximum possible number of danger zone violations that could occur, which is the number of time steps times $C^{N}_2$ ($N$ choose $2$).
\end{itemize}

\begin{figure}[]
  \centering
  \includegraphics[width=0.35\textwidth]{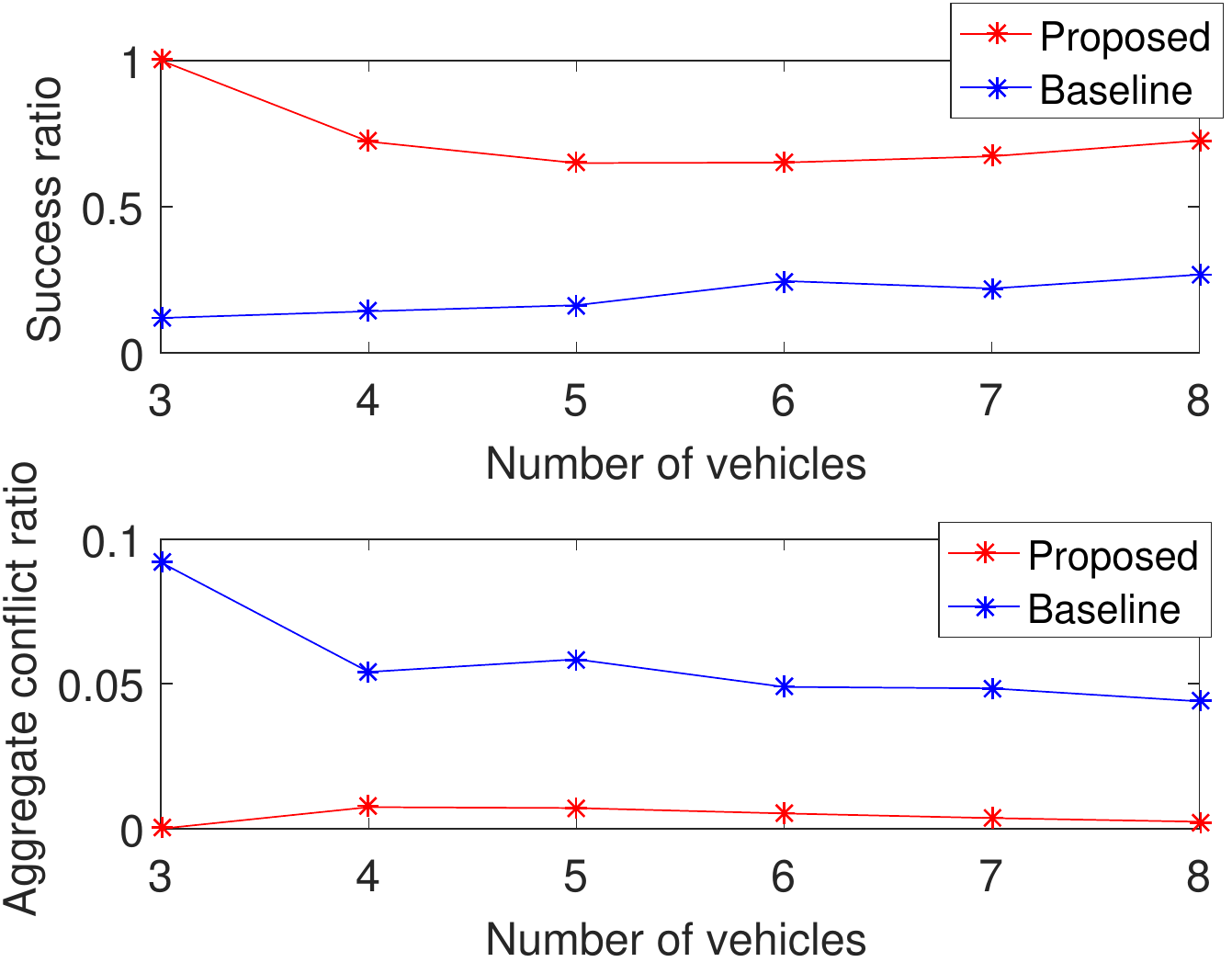}
  \caption{Our SAMV scheme outperforms the baseline method significantly in terms of success ratio and aggregate conflict ratio. In particular, we confirmed that for $N=3$, our method has a success ratio of $1.0$ and aggregate conflict ratio of $0.0$.}
  \label{fig:simulation_combined}
\vspace{-2em}
\end{figure}

With our proposed method, the average computation time per simulation is 4.1 seconds for $N=3$ and 25.5 seconds for $N=8$; this time includes the time needed to solve the MIP \eqref{eq:baseMIP}. With the baseline method, the average computation time for the same simulations is 5.9 seconds for $N=3$ and $22.9$ seconds for $N=8$. Both methods require the same BRS, which takes approximately 1 minute to compute. All computations were done on a MacBookPro 11.2 laptop with an Intel Core i7-4750 processor. 


\section{Conclusions}
By exploiting properties of pairwise optimal collision avoidance, our proposed mixed integer program method guarantees collision avoidance of three vehicle systems and performs well for larger multi-vehicle systems. 
\vspace{-0.5em}



 \bibliographystyle{IEEEtran}
 \bibliography{references}
\end{document}